\documentclass[a4paper,12pt,reqno]{article}

\usepackage[a4paper, left=2cm,right=2cm,top=2cm,bottom=2cm, includefoot,heightrounded]{geometry}
\usepackage{setspace}

\usepackage{mathpazo} 

\usepackage[hyphens]{url}
\usepackage[hyphenbreaks]{breakurl}
\usepackage{hyperref}
\hypersetup{colorlinks=true, linkcolor=blue, urlcolor  = blue,  citecolor=blue}

\usepackage{titlesec}
\usepackage[utf8]{inputenc}
\usepackage{graphicx}
\usepackage{algorithm}
\usepackage{algpseudocode}
\usepackage{amsmath}
\usepackage{MnSymbol}%
\usepackage{wasysym}%
\graphicspath{{Images/}}
\usepackage{amsfonts}
\usepackage{xcolor}

\usepackage[english]{babel}
\usepackage{amsthm}
\usepackage{tocloft}
\usepackage{bbm}
\usepackage{mathtools}
\usepackage{tikz}
\usepackage{comment}
\usetikzlibrary{decorations.pathreplacing}
\usetikzlibrary{positioning,quotes,calc}
\usepackage{subcaption}
\newtheorem*{theorem*}{Theorem}
\newtheorem*{corollary*}{Corollary}
\newtheorem*{proposition*}{Proposition}
\newtheorem*{lemma*}{Lemma}
\newtheorem*{fact*}{Fact}
\newtheorem*{definition*}{Definition}
\newtheorem*{conjecture*}{Conjecture}
\newtheorem{theorem}{Theorem}

\newtheorem{proposition}{Proposition}
\newtheorem{assumption}{Assumption}
\newtheorem{definition}{Definition}
\newtheorem{lemma}{Lemma}

\DeclareMathOperator{\T}{T}
\DeclareMathOperator{\argmax}{argmax}
\DeclarePairedDelimiterX{\inp}[2]{\langle}{\rangle}{#1, #2}

\titleformat{\section}
		{\large\bfseries\center}
         {\thesection}
        {0.5em}
        {}
        []
\titleformat{\subsection}
        {\normalfont\bfseries}
        {\thesubsection}
        {0.5em}
        {}
\titleformat{\subsubsection}[runin]
        {\normalfont\bfseries}
        {\thesubsubsection}
        {0.5em}
        {}
        [.]

\usepackage[authoryear]{natbib}
\usepackage{soul}
\title{Ironing Without Concavification\thanks{I am grateful to Ben Brooks, Piotr Dworczak and Ilya Segal for helpful comments.}}
\author{Filip Tokarski \\ Stanford GSB}

\begin{document}
\maketitle
\date{} 
\vspace*{-1cm} 

\begin{abstract}
I propose a new approach to solving standard screening problems when the monotonicity constraint binds. A simple geometric argument shows that when virtual values are quasi-concave, the optimal allocation can be found by appropriately truncating the solution to the relaxed problem. I provide an algorithm for finding this optimal truncation when virtual values are concave.
\end{abstract}

\section{Introduction}

This note revisits the problem of finding optimal menus in standard single-agent screening environments. When agents' payoffs are quasi-linear and satisfy single-crossing, the standard approach involves writing the objective as an integral of ``virtual values'' which depend only on one type's allocation and separately choosing each type's allocation to maximize its corresponding virtual value. Incentive compatibility, however, requires that the allocation be non-decreasing in type.
If this constraint binds, it is standard to transform virtual values so that after the transformation, pointwise maximization yields an increasing solution. This transformation, known as ironing, was described by \cite{myerson} and subsequently generalized by \cite{toikka}. A different approach uses optimal control methods \citep{guesnerie1984complete,hellwig2008maximum,ruiz2011non}.

I propose an alternative approach that involves solving the relaxed problem without the monotonicity constraint and transforming the resulting allocation to satisfy monotonicity. Theorem \ref{th:1} says that whenever an optimal allocation rule exists, it can be found by optimally truncating the solution to the relaxed problem. Moreover, the optimal truncation is pinned down by the allocations of types at which the solution to the relaxed problem changes monotonicity. These observations generalize insights about the structure of solutions from the literature---they require no continuity or differentiability assumptions, do not need virtual values to be concave, and do not assume that an agent's allocation is chosen from a compact interval. Therefore, in contrast to previous work, my results also apply when the planner can only assign discrete allocations.

I subsequently use these insights to develop a simple algorithm for finding the optimal allocation rule under the additional assumptions that virtual values are concave and that each agent's allocation is chosen from a compact interval. 

\section{Problem}\label{sec:probl}
Agents with types $\theta\in[0,1]$ are assigned allocations from a compact set $\mathcal X \subset \mathbb{R}$. The planner chooses a non-decreasing allocation rule $x:[0,1]\to \mathcal{X}$ to maximize:
\[
F[x]=\int_0^1 J(x(\theta),\theta) d \theta.
\]
I refer to $J : \mathcal{X}\times [0,1]\to \mathbb{R}$ as the {virtual value}.
\begin{assumption}\label{ass:1} The virtual value satisfies the following properties:
\begin{enumerate}
\item $J(\cdot, \theta)$ is weakly quasi-concave for every $\theta\in [0,1]$.
        \item $J(x,\theta)$ is uniformly bounded on $\mathcal{X}\times [0,1]$.
\end{enumerate}
\end{assumption}
I will call any $x$ that maximizes $J(\cdot,\theta)$ pointwise a solution to the {relaxed problem}. I also assume a well-behaved solution like that exists:
\begin{assumption}\label{ass:2}
There exists a solution to the relaxed problem, $x_R$, that is piecewise monotonic.
\end{assumption}

\section{Structure of the solution}\label{sec:struc}

In this section I present Theorem \ref{th:1} describing the structure of the solution to the planner's problem. I first introduce the following definition:
\begin{definition} A point $i\in (0,1)$ is a \textbf{critical point} if $x_R$ changes monotonicity there, that is, if $x_R$ is monotonic on $(i-\epsilon,i)$ for some $\epsilon>0$ but is not monotonic on $(i-\epsilon, i+\delta)$ for any $\delta>0$. By convention, I also call $0$ and $1$ critical points. I use $i_n$ to denote the $n$-th critical point and $\mathcal{I}$ to denote the set of critical points.
\end{definition}
I also define the function $x^*_v: [0,1]\to \mathcal{X}$ parametrized by $v=(v_1,\dots,v_{|\mathcal{I}|-1})\in \mathcal{X}^{|\mathcal I|-1}$:
        \begin{equation}\label{eq:tildex}
                x^*_v(\theta) =
                \begin{cases}
                \min \left\{ v_{n+1}, \max \left\{ v_{n}, x_R(\theta) \right\} \right\} & \text{if }\theta \in[i_n,i_{n+1}) \text{ where $x_R$ is increasing}, \\
                v_n & \text{if }\theta \in[i_n,i_{n+1}) \text{ where $x_R$ is decreasing},
        \end{cases}
        \end{equation}
where, by convention, $v_{|\mathcal{I}|}=x^*_v(1) = \max_{x\in\mathcal{X}} x$.
\begin{theorem}\label{th:1}
If the planner's problem has a solution, it has a solution of the form $x^*_v$ for some  $v\in \mathcal{X}^{|\mathcal I|-1}$. This solution can be recovered by solving:
        \begin{equation}\label{eq:simple}
                \max_{v \in \mathcal{X}^{|\mathcal I|-1}} F[x^*_v] \quad \text{subject to} \quad v_1 \leq v_2\leq \dots \leq v_{|\mathcal I|-1}.\tag{I}
        \end{equation}
        \end{theorem}
Theorem \ref{th:1} says that when the planner's problem has a solution, we can find it by optimally truncating the solution to the relaxed problem, $x_R$. Moreover, the optimal truncation is pinned down by the allocations of types at which $x_R$ changes monotonicity.

\paragraph{Proof of  Theorem \ref{th:1}.} I first prove three lemmas:
\begin{lemma}\label{fact:closer}
If $x_1,x_2$ are allocation rules and $x_2$ lies pointwise between $x_1$ and $x_R$, then $F[x_2]\geq F[x_1]$. 
\end{lemma}
\begin{proof}
$J(x_R(\theta),\theta) \geq J(x_1(\theta),\theta)$ for all $\theta$ by definition. Now, by quasi-concavity of $J(\cdot,\theta)$:
        \begin{align*}
        F[x_2] = \int_0^1 J(x_2(\theta),\theta) d \theta &\geq \int_0^1 \min\{J(x_1(\theta),\theta),J(x_R(\theta),\theta) \} d \theta\\
        & =  \int_0^1 J(x_1(\theta),\theta) d \theta = F[x_1].
        \end{align*}
\end{proof}
Lemma \ref{fact:closer} tells us that the objective always increases when we move the allocation rule pointwise closer to the first-best one. This property underpins the proofs of Lemmas \ref{lem:decreasing} and \ref{lem:increasing}, illustrated in Figures \ref{fig:subfig3} and \ref{fig:subfig4}.

\begin{lemma}\label{lem:decreasing} 
Suppose $x_R$ is decreasing on $[a,b)$. Then any increasing allocation rule $x$ can be weakly improved upon by some allocation rule $x^*$ that is constant on $[a,b)$ and coincides with $x$ elsewhere.
\end{lemma}
        \begin{proof}
Fix any increasing $x$. Consider $x^*$ that coincides with $x$ on $[0,1]\setminus [a,b)$ and takes the following values for $\theta \in [a,b)$:
\[
        x^*(\theta) = \begin{cases}
                x(a^+) & \text{if } x_R(\theta) \leq x(\theta) \text{ for all }\theta \in [a,b), \\
                x(b^-) & \text{if } x_R(\theta) \geq x(\theta) \text{ for all }\theta \in [a,b),\\
                \max\{ x(t^-),x_R(t^+) \}  & \text{otherwise},
                      \end{cases}
\]
where $t:=\sup\{\theta\in(a,b): x_R(\theta)\ge x(\theta)\}$. Note $x^*$ is increasing and pointwise between $x_R$ and $x$, so $F[x^*] \geq F[x]$ by Lemma \ref{fact:closer}.
        \end{proof}

\begin{lemma}\label{lem:increasing}
Suppose $x_R$ is increasing on $[a,b)$. Then any increasing allocation rule $x$ can be weakly improved upon by:
\[
x^*(\theta)=
\begin{cases}
\min\bigl\{x(b),\,\max\{x(a),\,x_R(\theta)\}\bigr\}, & \theta\in[a,b),\\[4pt]
x(\theta), & \theta\notin[a,b).
\end{cases}
\]
\end{lemma}
\begin{proof}
Fix any increasing $x$. Note $x^*$ is pointwise between $x$ and $x_R$, so $F[x^*]\geq F[x]$ by Lemma \ref{fact:closer}. Moreover, $x^*$ is increasing because $x_R$ is increasing on $[a,b)$ by assumption.
\end{proof}

\begin{figure}[h!]
  \centering
  \begin{subfigure}{0.4\linewidth}
        \begin{tikzpicture}[font=\footnotesize,scale=0.85]
                \draw[color=red,ultra thick] (0,2) -- (6,2);
                \draw[color=teal,thick] (0,0.75) to[out=20, in=-140] (6,4);
                \draw[color=blue,thick] (0,3.75) to[out=-20, in=180] (6,1);
                \draw(5.9,2.8) node[label=above:{$x(\theta)$}] {};
                \draw(5.9,1) node[label=above:{$x^*(\theta)$}] {};
                \draw(5.9,0.1) node[label=above:{$x_R(\theta)$}] {};
                \draw[thick, dashed] (3,2) -- (3,0);
                \draw[->,thick] (0,0) -- (6,0);
                \draw (0,0) node[anchor=north] {$a$};
                \draw (6,0) node[anchor=north] {$b$};
                \draw (3,0) node[anchor=north] {$\sup$};
                \draw (3,-0.4) node[anchor=north] {$\{\theta \in (a,b) : x_R(\theta) \geq x(\theta) \}$};
                \draw[->,thick] (0,0) -- (0,4.5);
                \end{tikzpicture}
                \captionsetup{justification=centering} 
      \caption{Figure \ref{fig:subfig3}: Constructing an improvement in the proof of Lemma \ref{lem:decreasing}.}
      \label{fig:subfig3}
  \end{subfigure}
  \ \ \ \ \ \  \ 
  \begin{subfigure}{0.4\linewidth}
        \begin{tikzpicture}[font=\footnotesize, scale=0.85]
                \draw[thick, dashed] (0.6,0) -- (0.6,0.8);
                \draw[thick, dashed] (5.15,0) -- (5.15,3.03);
                \draw (0.6,0) node[anchor=north] {$a$};
                \draw (5.15,0) node[anchor=north] {$b$};
                \draw[color=blue,thick] (0,0.2) to[out=50, in=-120] (6,4);
                \draw[color=teal,thick] (0,0.8) to[out=50, in=-120] (6,3);
                \draw[color=red,ultra thick] (0.6,0.8) to[out=38, in=-142] (5.17,3);
                \draw[color=red, ultra thick] (0,0.8) -- (0.63,0.8); 
                \draw[color=red, ultra thick] (5.15,3) -- (6,3); 
                \draw(3.8,2.2) node[label=above:{$x^*(\theta)$}] {};
                \draw(5,3.1) node[label=above:{$x_R(\theta)$}] {};
                \draw(4.5,1) node[label=above:{$x(\theta)$}] {};
                \draw[->,thick] (0,0) -- (6,0);
                \draw[->,thick] (0,0) -- (0,4.5);
            \end{tikzpicture}
            \vspace*{0.437cm} 
                \captionsetup{justification=centering} 
      \caption{Figure \ref{fig:subfig4}: Constructing an improvement in the proof of Lemma \ref{lem:increasing}.}
      \label{fig:subfig4}
  \end{subfigure}
\end{figure}
Observations similar to Lemmas \ref{lem:decreasing} and \ref{lem:increasing} were made by \cite{sandmann2022single} who studies the optimality of sparse menus in a price discrimination problem. Given the above lemmas, the proof of Theorem \ref{th:1} is straightforward. By Assumption \ref{ass:2}, $[0,1)$ can be partitioned into finitely many intervals $[i_n,i_{n+1})$ where $x_R$ is monotonic. Then, by Lemmas \ref{lem:decreasing} and \ref{lem:increasing}, if the planner's problem has a solution, it has one of the form in \eqref{eq:tildex} with $v_1 \leq v_2\leq \dots \leq v_{|\mathcal I|-1}$. Solving problem \eqref{eq:simple} recovers this solution.

\section{Solution algorithm}\label{sec:alg}
I now provide a simple algorithm that solves planner's problem under the following additional assumption:
\begin{assumption}\label{ass:reg} 
The virtual value satisfies the following properties:
\begin{enumerate}
        \item The planner chooses allocations from a closed interval: $\mathcal{X}=[l,h]$.
        \item $J(\cdot,\theta)$ is continuous and weakly concave for every $\theta\in [0,1]$. 
\end{enumerate} 
\end{assumption}
The algorithm uses the following transformation $\T:[l,h]^{[0,1]}\times [l,h] \times \mathbb{N}\to {[l,h]}^{[0,1]}$:
\[
        \T[x,\tilde v,n](\theta) :=
        \begin{cases}
        \min \{\tilde v,x(\theta)\} & \text{if }\theta < i_n,\\
        \tilde v & \text{if }\theta\in[i_n,i_{n+1}) \text{ where $x_R$ is decreasing}, \\
        \max \left\{\tilde v, x(\theta)\right\} &  \text{if }\theta\in[i_n,i_{n+1}) \text{ where $x_R$ is increasing}, \\
        h& \text{if $\theta=i_{n+1}$},\\
        x(\theta) & \text{if }\theta > i_{n+1}.
        \end{cases}
\]
When $\T[\cdot \ ,v,n]$ is applied to $x$, the allocation rule is truncated from above by $\tilde v$ before the $n$th critical point and set equal to $\tilde v$ or truncated by it from below between the $n$th and $(n+1)$st critical points. It is also set to $h$ at the $(n+1)$st critical point. The following algorithm takes in the solution to the relaxed problem, $x_R$, and the set of critical points $i\in \mathcal{I}$ as inputs and iteratively applies this transformation to produce a solution to the planner's problem:

\begin{algorithm}[H]\caption{}
        \setstretch{1.1}\label{alg:1}
        \begin{algorithmic}
            \State $r \gets x_R$
            \State $n\gets 1$
            \While{$n\leq |\mathcal{I}|-1$}
                \State $v^* \gets \argmax_{\tilde v\in [l,h]} \int_0^{i_{n+1}} J (\T[r,\tilde v,n](\theta),\theta ) d\theta$
                \State $r \gets \T[r,v^*,n]$ 
                \State $n\gets n+1$
                \EndWhile
                \State \Return $r$
\end{algorithmic}
\end{algorithm}

Intuitively, each iteration of Algorithm \ref{alg:1} starts with the optimum subject to monotonicity from 0 until the $n$th critical point and ``irons out'' this partial solution further to produce a solution subject to monotonicity until the $(n+1)$st critical point (Figure \ref{fig:algo}).

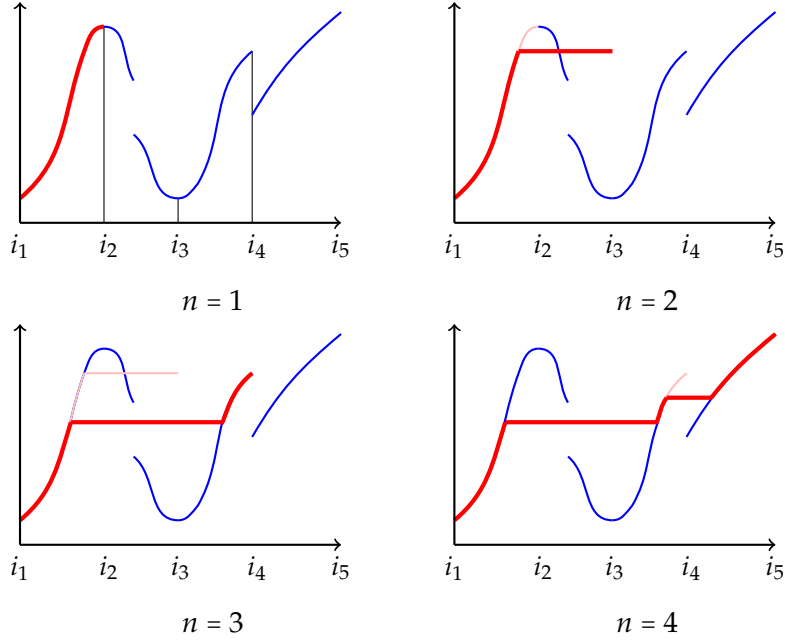
\begin{figure}[h!]
        \centering
        \begin{subfigure}{0.33\linewidth}
              \begin{tikzpicture}[font=\footnotesize,scale=0.65]
                  \draw (0,0.5) node[anchor=north] {$i_1$};
                  \draw (1.8,0.5) node[anchor=north] {$i_2$};
                  \draw (1.7,4.5) -- (1.7,0.5);
                  \draw (3.25,0.5) node[anchor=north] {$i_3$};
                  \draw (3.2,1) -- (3.2,0.5);
                  \draw (4.8,0.5) node[anchor=north] {$i_4$};
                  \draw (4.7,4) -- (4.7,0.5);
                  \draw (6.5,0.5) node[anchor=north] {$i_5$};
                  \draw[color=red,ultra thick] (0,1) to [out=40, in=250]  (1.3,4) to [out=70, in=180] (1.7,4.5);
                  \draw[color=blue,thick] (1.7,4.5) to[out=360, in=120] (2.3,3.4);
                  \draw[color=blue,thick] (2.3,2.3) to[out=-40, in=180] (3.2,1);
                  \draw[color=blue,thick] (3.2,1) to[out=0, in=-130] (3.6,1.3) to[out=60, in=-140] (4.7,4);
                  \draw[color=blue,thick] (4.7,2.7) to[out=60, in=-140] (6.5,4.8);
                  \draw[->,thick] (0,0.5) -- (6.5,0.5);
                  \draw[->,thick] (0,0.5) -- (0,5);
              \end{tikzpicture}
            \caption{$n=1$}
        \end{subfigure}
        \begin{subfigure}{0.33\linewidth}
              \begin{tikzpicture}[font=\footnotesize,scale=0.65]
                      \draw (0,0.5) node[anchor=north] {$i_1$};
                      \draw (1.8,0.5) node[anchor=north] {$i_2$};
                      \draw (3.25,0.5) node[anchor=north] {$i_3$};
                      \draw (4.8,0.5) node[anchor=north] {$i_4$};
                      \draw (6.5,0.5) node[anchor=north] {$i_5$};
                      \draw[color=pink,thick] (0,1) to [out=40, in=250]  (1.3,4) to [out=70, in=180] (1.7,4.5);
                      \draw[color=blue,thick] (1.7,4.5) to[out=360, in=120] (2.3,3.4);
                      \draw[color=blue,thick] (2.3,2.3) to[out=-40, in=180] (3.2,1);
                      \draw[color=blue,thick] (3.2,1) to[out=0, in=-130] (3.6,1.3) to[out=60, in=-140] (4.7,4);
                      \draw[color=blue,thick] (4.7,2.7) to[out=60, in=-140] (6.5,4.8);
                      \draw[color=red,ultra thick] (0,1) to [out=40, in=250]  (1.3,4);
                      \draw[color=red,ultra thick] (1.3,4) -- (3.2,4);
                      \draw[->,thick] (0,0.5) -- (6.5,0.5);
                      \draw[->,thick] (0,0.5) -- (0,5);
                  \end{tikzpicture}
            \caption{$n=2$}
        \end{subfigure}
        \begin{subfigure}{0.33\linewidth}
              \begin{tikzpicture}[font=\footnotesize,scale=0.65]
                  \draw (0,0.5) node[anchor=north] {$i_1$};
                  \draw (1.8,0.5) node[anchor=north] {$i_2$};
                  \draw (3.25,0.5) node[anchor=north] {$i_3$};
                  \draw (4.8,0.5) node[anchor=north] {$i_4$};
                  \draw (6.5,0.5) node[anchor=north] {$i_5$};
                  \draw[color=blue,thick] (0,1) to [out=40, in=250]  (1.3,4) to [out=70, in=180] (1.7,4.5);
                  \draw[color=blue,thick] (1.7,4.5) to[out=360, in=120] (2.3,3.4);
                  \draw[color=blue,thick] (2.3,2.3) to[out=-40, in=180] (3.2,1);
                  \draw[color=blue,thick] (3.2,1) to[out=0, in=-130] (3.6,1.3) to[out=60, in=-140] (4.7,4);
                  \draw[color=blue,thick] (4.7,2.7) to[out=60, in=-140] (6.5,4.8);
                  \draw[color=pink, thick] (0,1) to [out=40, in=250]  (1.3,4);
                  \draw[color=pink, thick] (1.3,4) -- (3.2,4);
                  \draw[color=red,ultra thick] (0,1) to [out=40, in=253]  (1.03,3);
                  \draw[color=red,ultra thick] (4.1,3) to [out=70, in=-140] (4.7,4);
                  \draw[color=red,ultra thick] (1.028,3) -- (4.13,3);
                  \draw[->,thick] (0,0.5) -- (6.5,0.5);
                  \draw[->,thick] (0,0.5) -- (0,5);
              \end{tikzpicture}
            \caption{$n=3$}
        \end{subfigure}
        \begin{subfigure}{0.33\linewidth}
              \begin{tikzpicture}[font=\footnotesize,scale=0.65]
                      \draw (0,0.5) node[anchor=north] {$i_1$};
                      \draw (1.8,0.5) node[anchor=north] {$i_2$};
                      \draw (3.25,0.5) node[anchor=north] {$i_3$};
                      \draw (4.8,0.5) node[anchor=north] {$i_4$}; 
                      \draw (6.5,0.5) node[anchor=north] {$i_5$};
                      \draw[color=blue,thick] (0,1) to [out=40, in=250]  (1.3,4) to [out=70, in=180] (1.7,4.5);
                      \draw[color=blue,thick] (1.7,4.5) to[out=360, in=120] (2.3,3.4);
                      \draw[color=blue,thick] (2.3,2.3) to[out=-40, in=180] (3.2,1);
                      \draw[color=blue,thick] (3.2,1) to[out=0, in=-130] (3.6,1.3) to[out=60, in=-110] (4.12,3);
                      \draw[color=blue,thick] (4.7,2.7) to[out=60, in=-140] (6.5,4.8);
                      \draw[color=pink,thick] (4.07,3) to [out=70, in=-140] (4.7,4);
                      \draw[color=red,ultra thick] (0,1) to [out=40, in=253]  (1.03,3);
                      \draw[color=red,ultra thick] (4.1,3) to [out=70, in=-127] (4.3,3.5);
                      \draw[color=red,ultra thick] (1.028,3) -- (4.13,3);
                      \draw[color=red,ultra thick] (4.298,3.5) -- (5.22,3.5);
                      \draw[color=red,ultra thick] (5.2,3.5) to[out=52, in=-140] (6.5,4.8);
                      \draw[->,thick] (0,0.5) -- (6.5,0.5);
                      \draw[->,thick] (0,0.5) -- (0,5);
                  \end{tikzpicture}
            \caption{$n=4$}
        \end{subfigure}
        \caption{Algorithm \ref{alg:1} recursively transforming $x_R$ (blue) into subsequent $\T[x,v^*,n]$ (red).}\label{fig:algo}
      \end{figure}
\begin{theorem}\label{th:2}
Under Assumption \ref{ass:reg}, the output of Algorithm \ref{alg:1} solves the planner's problem.
\end{theorem}
\paragraph{Proof of Theorem \ref{th:2}.} For $a\in [0,1)$, define the partial objective:
\[
        F_a[x]:= \int_0^a J(x(\theta),\theta) d\theta.
\]
Now, consider the set of allocation rules that are increasing on $[0,a)$ and below $k$ everywhere on that interval. Let $\mathcal{O}^k_a$ contain the allocation rules that maximize $F_a$ over this set:
\[
\mathcal O_a^k
:= \arg\max\Bigl\{F_a[x] \ \ \text{s.t.} \ \ x:[0,1]\to [l,h]  \text{ is increasing on }[0,a),\ x\le k \text{ on }[0,a)\Bigr\}.
\]
I first show $\mathcal O_a^k$ is non-empty for all $k\in [l,h]$ and $a\in (0,1]$. Let $\mathcal A_a^k$ be the set of increasing functions $x:[0,a]\to[l,k]$. By Helly's selection theorem, $\mathcal A_a^k$ is sequentially compact under pointwise convergence.
If $x_m\to x$ pointwise in $\mathcal A_a^k$, then $J(x_m(\theta),\theta)\to J(x(\theta),\theta)$ for each $\theta$.
By Assumption \ref{ass:1}, $|J(x_m(\theta),\theta)|\le M$ uniformly, so dominated convergence gives
$F_a[x_m]\to F_a[x]$. Hence $F_a$ attains its maximum on $\mathcal A_a^k$.

I also prove the following proposition. It says that when we have an $x$ in $\mathcal{O}_a^h$ and impose the additional constraint that $x$ be below $k$ on $[0,a)$, we need not resolve the problem, but can simply truncate the solution without this constraint by $k$.
\begin{proposition}\label{lem:truncation}
        If $x_a\in \mathcal{O}^{h}_a$, then $\min\{k, x_a\}\in \mathcal{O}^k_a$.
        \end{proposition}        
\begin{proof}
Fix $x_a$ and note $\min\{x_a(\theta),k\}$ is admissible in the problem defining $\mathcal{O}^k_a$. Fix any other admissible $x$ and define:
\[
\tilde x(\theta)
=
\begin{cases}
x(\theta)+\max\{x_a(\theta)-k,\,0\}, & \theta\in[0,a),\\
h, & \text{otherwise}.
\end{cases}
\]
Since $x$ and $\max\{x_a(\theta)-k, \ 0 \}$ are non-decreasing, so is $\tilde x$. 
Moreover, $x(\theta)\le k$ for $\theta<a$ and $x_a(\theta)\le h$, so $\tilde x(\theta)\leq h$. Now, $x_a\in \mathcal{O}^{h}_a$, and so:
\begin{equation}\label{eq:opt_xstar}
\int_0^a J(x_a(\theta),\theta)\,d\theta \;\ge\;\int_0^a J(\tilde x(\theta),\theta)\,d\theta.
\end{equation}
Fix $\theta$ and consider two cases. If $x_a(\theta)\le k$, then $\min\{x_a(\theta),k\}=x_a(\theta)$ and $\tilde x(\theta)=x(\theta)$, hence:
\begin{equation}\label{eq:case1}
J(\min\{x_a(\theta),k\},\theta)-J(x(\theta),\theta)
=
J(x_a(\theta),\theta)-J(\tilde x(\theta),\theta).
\end{equation}
Now suppose $x_a(\theta)>k$. Let $\delta : = k-x(\theta)\geq 0$ and note that:
\[
\min\{x_a(\theta),k\} = x(\theta) + \delta.
\]
Since $J(\cdot,\theta)$ is concave, for all $u,u'$ such that $u<u'$:
\[
J(u+\delta,\theta)-J(u,\theta) \ \geq J(u'+\delta,\theta)-J(u',\theta).
\]
Recall $ x(\theta) <\tilde x(\theta)$, so we have:
\[
J(x(\theta)+\delta,\theta)-J(x(\theta),\theta)\;\ge\;J(\tilde x(\theta)+\delta,\theta)-J(\tilde x(\theta),\theta).
\]
Note $x(\theta) + \delta = k = \min\{x_a(\theta),k\}$ and $\tilde x(\theta)+\delta = x(\theta) + x_a(\theta) - k + (k-x(\theta)) = x_a(\theta)$, so:
\begin{equation}\label{eq:case2}
J(\min\{x_a(\theta),k\},\theta)-J(x(\theta),\theta)\;\ge\;J(x_a(\theta),\theta)-J(\tilde x(\theta),\theta).
\end{equation}
Integrating over $\theta\in[0,a]$ and combining \eqref{eq:opt_xstar}, \eqref{eq:case1} and \eqref{eq:case2} yields: 
\[
\int_0^a \bigl[J(\min\{x_a(\theta),k\},\theta)-J(x(\theta),\theta)\bigr]\,d\theta
\;\ge\;
\int_0^a \bigl[J(x_a(\theta),\theta)-J(\tilde x(\theta),\theta)\bigr]\,d\theta
\;\ge\;0,
\]
giving $\int_0^a J(\min\{x_a(\theta),k\},\theta)\,d\theta \ge \int_0^a J(x(\theta),\theta)\,d\theta$.
\end{proof}
I now present an inductive proof of Theorem \ref{th:2}. Recall that $|\mathcal{I}|\geq 2$ since $0,1\in \mathcal{I}$. The base case demonstrates that the first iteration of the algorithm produces $r\in \mathcal{O}^{h}_{i_{2}}$. The step shows that when the $n$th iteration starts with $r\in \mathcal{O}^{h}_{i_n}$, it produces $r\in \mathcal{O}^{h}_{i_{n+1}}$. The base and the step thus imply that the algorithm will produce $r\in \mathcal{O}^{h}_{i_{|\mathcal{I}|}}$ solving the planner's problem in $|\mathcal{I}|-1$ steps.

        \textbf{Base.} Since $\mathcal{O}_{i_2}^{h}\neq \emptyset$, Lemmas \ref{lem:decreasing} and \ref{lem:increasing} tell us there exists $x^*\in \mathcal{O}_{i_2}^{h}$ for which:
        \[
                x^*(\theta) :=
                \begin{cases}
                v_1 & \text{if }\theta \in[i_1,i_{2}) \text{ where $x_R$ is decreasing}, \\
                \min \left\{ v_2, \max \left\{ v_{1}, x_R(\theta) \right\} \right\} &  \text{if }\theta \in[i_1,i_{2}) \text{ where $x_R$ is increasing}.
                \end{cases}     
        \]
        Moreover, by Lemma \ref{fact:closer} we can without loss set $v_2=h$. The first iteration of the algorithm will recover such an $x^*$ when optimizing over $\tilde v$.
        
        \textbf{Step.} Since $\mathcal{O}^{h}_{i_{n+1}}\neq \emptyset$, Lemmas  \ref{lem:decreasing} and \ref{lem:increasing} tell us there exists $x^*\in \mathcal{O}_{i_{n+1}}^{h}$ for which:
        \[
                x^*(\theta) :=
                \begin{cases}
                v_n & \text{if }\theta \in[i_n,i_{n+1}) \text{ where $x_R$ is decreasing}, \\
                \min \left\{ v_{n+1}, \max \left\{ v_{n}, x_R(\theta) \right\} \right\} &  \text{if }\theta \in[i_n,i_{n+1}) \text{ where $x_R$ is increasing},
                \end{cases}     
        \]
        where $x^*(i_n)=v_n$. Moreover, by Lemma \ref{fact:closer} we can without loss set  $v_{n+1}=h$. Let $r\in \mathcal{O}^{h}_{i_n}$ be the allocation rule produced by the $n-1$st iteration of the algorithm. By Proposition \ref{lem:truncation}, $\min\{v_n, r\}\in \mathcal{O}^{v_n}_{i_n}$. Since $x^*(\theta)\leq v_n$ for $\theta<i_n$, this implies $F_{i_n}[\min\{v_n, r\}] \geq F_{i_n}[x^*]$ and so the following allocation rule weakly improves upon $x^*$:
        \[
                x^{**}(\theta) :=
                \begin{cases}
                \min \{v_n,r(\theta)\} & \text{if }\theta < i_n,\\
                v_n & \text{if }\theta \in[i_n,i_{n+1}) \text{ where $x_R$ is decreasing}, \\
                \max \left\{v_{n}, x_R(\theta)\right\} &  \text{if }\theta \in[i_n,i_{n+1}) \text{ where $x_R$ is increasing}.
                \end{cases}        
        \]
        Moreover, $x^{**}$ is increasing on $[0,i_{n+1})$ and takes values in $[l,h]$, so $x^{**}$ belongs to $\mathcal{O}_{i_{n+1}}^{h}$.
        The $n$th iteration of the algorithm will recover such an $x^{**}$ when optimizing over $\tilde v$.


\end{document}